\documentclass[11pt]{article}
\usepackage[T1]{fontenc}
\usepackage{lmodern}
\usepackage{amsmath,amssymb,amsthm,mathtools}
\usepackage[margin=1.05in]{geometry}
\usepackage{xcolor}
\usepackage[protrusion=true,expansion=true]{microtype}
\usepackage[colorlinks=true,linkcolor=blue!45!black,citecolor=blue!45!black,urlcolor=blue!45!black]{hyperref}
\numberwithin{equation}{section}
\newtheorem{theorem}{Theorem}[section]
\newtheorem{lemma}[theorem]{Lemma}
\newtheorem{proposition}[theorem]{Proposition}
\newcommand{\Ent}{\operatorname{Ent}}
\newcommand{\E}{\mathbb E}
\newcommand{\PP}{\mathbb P}
\hypersetup{pdftitle={Optimal Mixing for Local Edge-Coloring Dynamics on Trees},pdfauthor={Haoran Wang}}
\title{Optimal Mixing for Local Edge-Coloring Dynamics on Trees}
\author{Haoran Wang\\\texttt{whr.hrwang@gmail.com}}
\date{}
\begin{document}
\maketitle
\begin{abstract}
We study single-edge Glauber dynamics and neighboring-edge heat-bath dynamics for uniformly sampling proper edge colorings of a tree. We obtain log-Sobolev estimates that imply $O_{\Delta,q}(n\log n)$ mixing when $q\ge\Delta+2$ for single-edge updates and $q\ge\Delta+1$ for neighboring-edge updates. The proof uses local color release with a doubly exponential failure tail and a fixed preparation that permits arbitrary interior recoloring paths while restoring the exterior. When one spare color is available, neighboring-edge updates provide a second-layer controller in place of the two-spare-color buffer. Tree-message concentration and component entropy factorization then turn this typical connectivity into a functional inequality for the original local chains.
\end{abstract}

\section{Introduction}
Proper edge colorings of a tree can be sampled exactly by proceeding from a root to the leaves. The behavior of local Markov chains is less straightforward. Recoloring one edge must respect two endpoint constraints, and a move that is legal after a simultaneous color exchange may not be realizable by a short sequence of local updates. The problem is especially natural when the number of colors is close to the maximum degree $\Delta$.

Delcourt, Heinrich, and Perarnau~\cite{DHP20} proved polynomial mixing for the single-edge Glauber chain when $q\ge\Delta+1$, its irreducibility threshold on trees. On paths this threshold can exhibit much slower mixing~\cite{DGJ06}. Carlson, Chen, Feng, and Vigoda~\cite{CCFV26} proved an $O_{\Delta,q}(n)$ relaxation time for single-edge Glauber dynamics when $q\ge\Delta+2$, and the same bound for neighboring-edge heat-bath dynamics when $q\ge\Delta+1$. For single-edge Glauber they also obtained an $O_{\Delta,q}(nD\log n)$ mixing bound on trees of diameter $D$, and an $O_{\Delta,q}(n\log^2 n)$ bound on complete regular trees. The optimal mixing time on arbitrary trees remained open.

Chen, Wang, Zhang, and Zhang~\cite{CWZZ25} proved that weak spatial mixing holds for edge colorings of every tree of maximum degree $\Delta$ precisely when $q\ge 2\Delta-1$. The threshold $q=\Delta+2$ for single-edge dynamics lies below this range when $\Delta\ge4$. Our proof instead controls local connectivity and posterior information under the free-boundary measure.

\subsection{Main results}
For a tree $T$, let $\Omega=\Omega_q(T)$ denote its proper $q$-edge colorings and let $\pi$ be the uniform law on $\Omega$. The neighboring-edge chain selects an unordered pair $\{e,f\}$ with $e=f$ or $e\cap f\ne\varnothing$ and uniformly resamples its colors conditional on the other edges. Single-edge Glauber dynamics selects an edge and proposes a uniformly random color, accepting the proposal when legal. Both chains are reversible with stationary distribution $\pi$.

For $g:\Omega\to\mathbb R$, write
\[
 \mathcal D_1(g)=\sum_{e}\pi[\operatorname{Var}_e(g)],\qquad
 \mathcal D_2(g)=\sum_{\{e,f\}:\,e=f\text{ or }e\cap f\ne\varnothing}\pi[\operatorname{Var}_{\{e,f\}}(g)].
\]
The conditional variances are taken under $\pi$ after fixing every color outside the selected edges.

\begin{theorem}\label{thm:main}
Fix $\Delta\ge2$ and $q\ge\Delta+1$. Every finite tree $T$ on $n$ vertices and of maximum degree at most $\Delta$ satisfies
\begin{equation}\label{eq:main-lsi}
 \operatorname{Ent}_\pi(f)\le C_{\Delta,q}\mathcal D_2(\sqrt f)
 \qquad(f\ge0).
\end{equation}
If $q\ge\Delta+2$, then also
\[
 \operatorname{Ent}_\pi(f)\le C_{\Delta,q}\mathcal D_1(\sqrt f).
\]
Moreover, for each fixed $\Delta$ there is a constant $K_\Delta$ independent of $q$ such that each chain, at its respective threshold, satisfies
\begin{equation}\label{eq:uniform-mix}
 t_{\mathrm{mix}}(\varepsilon)
 \le K_\Delta n\left(\log n+\log\frac1\varepsilon\right),
 \qquad 0<\varepsilon\le\tfrac12.
\end{equation}
\end{theorem}

The order $n\log n$ is sharp for both chains. The bounded-block extension of the Hayes--Sinclair lower bound~\cite{HS07} applies to neighboring-edge updates as well as single-edge updates. The palette thresholds differ by one because two unused colors allow a single-edge buffer, whereas one unused color requires a neighboring-edge move to transfer an obstruction into a second layer.

\subsection{Proof overview}
An interior recoloring path may visit the same interface edge repeatedly. Preparing its exterior separately at each visit would not give a uniform bound. We therefore prepare a fixed region below the interface, whose state returns to the prepared configuration after each interface move. The same buffer then supports an arbitrary recoloring path.

For single-edge updates, two unused colors support a three-color buffer. With one unused color, a neighboring-edge update can exchange the color on an interface child edge with the color on its own child. A second-layer controller gives the same reusable property. Under the free-boundary measure, failure of a prescribed release forces independent failures in distinct descendant subtrees, giving a doubly exponential tail.

The entropy argument uses the component-dynamics framework of Sly and Zhang~\cite{SZ17}. A local reconfiguration component may convey information about the root even after the block exterior is fixed. The buffer construction makes this component label constant on an entire interior fiber once a typical intermediate cutset is given. Tree-message concentration then controls the root posterior, and component entropy factorization yields the log-Sobolev bounds. The component graphs use the original local updates.

Local recoloring and canonical paths for tree edge colorings were studied in~\cite{DHP20,CCFV26}. Sly and Zhang~\cite{SZ17} developed component dynamics for tree vertex colorings, including the preparation and reversal of exterior recolorings. Related results for WSK dynamics and line-graph Glauber dynamics appear in~\cite{WYZ26,Haqi26}.

\section{Preliminaries}
Fix a root of $T$, direct every edge away from it, and order the outgoing edges at each vertex. Given the color $a$ of the incoming edge at a nonroot vertex $v$, its $d_v\le\Delta-1$ outgoing edge colors form a uniform injective tuple from $[q]\setminus\{a\}$. Conditional on this tuple, the descendant subtrees are independent. This follows directly from the number of colorings extending any prescribed tuple. The colors missing at $v$ form the set
$[q]\setminus\{\sigma(e):e\ni v\}$. At $q=\Delta+1$ there is at least
one such color; at $q=\Delta+2$ there are at least two.

A local recoloring path is a sequence of proper colorings, each obtained from the preceding one by changing one edge or two incident edges. When two edges are chosen, either may retain its color. Every step described below therefore has positive transition probability under the appropriate heat-bath chain. Single-edge moves are also allowed in neighboring-edge dynamics because $e=f$ is an admissible pair.

For a nonnegative function $f$, we use $\Ent_\pi(f)=\pi(f\log f)-\pi(f)\log\pi(f)$. The sums defining $\mathcal D_1,\mathcal D_2$ have $\Theta_\Delta(n)$ terms. An unnormalized log-Sobolev bound with a constant depending on $\Delta,q$ therefore gives an $O_{\Delta,q}(n\log n)$ mixing bound for the normalized local chain.

\section{Local recoloring}\label{sec:local}

\subsection{Color release}
For an edge entering a vertex $v$, its color is held fixed throughout a
release. A color $b$ is \emph{released at $v$} if no edge directed out of $v$
uses $b$. The depth of a release is the greatest distance below $v$ of an
edge modified by the procedure. All probabilities in this subsection are
conditional on the original incoming color and the entire original outgoing
star at $v$. Under the uniform coloring law, the unexposed descendant
subtrees then have independent distributions.

\begin{lemma}\label{lem:release}
Suppose a prescribed color $b$ differs from the fixed incoming color of $v$.
For $q=\Delta+2$, the probability that single-edge updates cannot release $b$
within depth $r\ge1$ is at most $\beta^{2^{r-1}}$, where
\[
 \beta=\frac{(\Delta-1)(\Delta-2)}{\Delta(\Delta+1)}.
\]
For $q=\Delta+1$, neighboring-edge updates release $b$ within depth $2k$
except with probability at most $\gamma^{2^k-1}$, where
\[
 \gamma=\frac{(\Delta-1)(\Delta-2)}{\Delta^2},\qquad k\ge1.
\]
Both estimates hold uniformly over the tree, the original star, and $b$.
For $\Delta=2$ the corresponding release succeeds deterministically within
the stated initial depth.
\end{lemma}
\begin{proof}
If $b$ is absent from the outgoing star, no update is necessary. Otherwise,
write $vw$ for its unique outgoing edge. In the single-edge case choose two
distinct colors $c_1,c_2$ missing at $v$. If one is missing at $w$, use it
to recolor $vw$. If both occur at $w$, they occupy two distinct outgoing
edges. Conditional on the color $b$ of $vw$, the outgoing colors at $w$ are
a uniform injection from $[q]\setminus\{b\}$. The probability that both
$c_1,c_2$ occur is at most
\[
 \frac{(\Delta-1)(\Delta-2)}{(q-1)(q-2)}=\beta.
\]
On this event, the two occupied edges lead into disjoint descendant subtrees.
Release $c_1$ below the first edge or $c_2$ below the second; either
successful release enables the final recoloring of $vw$. Each recursive
release keeps its own incoming edge fixed, so both searches may be made on
the original configuration. Conditional on the full star at $w$, their
success events are independent, and each failure probability is at most
the worst depth-$(r-1)$ failure probability. Thus $p_1\le\beta$ and
$p_r\le p_{r-1}^2$ for $r\ge2$. Iteration proves the first bound.

For the neighboring-edge chain, choose $c$ missing at $v$. If it is also
missing at $w$, recolor $vw:b\to c$. Otherwise the unique edge $wy$ with
color $c$ can be updated together with $vw$. Choose $t$ missing at $w$.
The possible updates are
\begin{equation}\label{eq:pair-options}
 (vw,wy):(b,c)\longrightarrow(c,b)
 \quad\text{or}\quad
 (vw,wy):(b,c)\longrightarrow(c,t).
\end{equation}
The first is legal if $b$ is absent below $y$; the second is legal if $t$
is absent there. If both are blocked, the two colors occur on different
outgoing edges of $y$. Releasing either one below $y$, while keeping $wy$
fixed, allows the corresponding update in \eqref{eq:pair-options}.
Conditional on the original star at $v$, the chance that $c$ occurs below
$w$ is at most $(\Delta-1)/\Delta$. Conditional on the star at $w$, the
chance that $b,t$ both occur below $y$ is at most $(\Delta-2)/\Delta$.
These two conditional bounds and the independence of the two descendant
release attempts give
\[
 p_1\le\gamma,\qquad p_{k+1}\le\gamma p_k^2.
\]
All searches are specified from the original coloring, and their updates
are executed only after a successful branch has been selected. The
recursion therefore does not require independence of the configurations
after modification. The stated bound follows by induction.
\end{proof}

The use of two releases at the same vertex requires a separate observation.
It is needed when an edge must accept two different future colors, which
is not guaranteed by applying the preceding lemma independently and then
performing the resulting paths in an arbitrary order.

\begin{lemma}\label{lem:simultaneous}
Let $q\ge\Delta+2$. Fix the incoming edge and the original outgoing star at
$v$, and fix a set $S$ of at most two colors, none equal to the incoming
color. There is a deterministic choice of release procedure, based on the
original coloring, that makes every color of $S$ absent from the outgoing
star without changing the incoming edge. It uses only edges within depth
$r\ge2$ and fails with conditional probability at most
$2\beta^{2^{r-2}}$. The procedure is reversible when the outgoing edges
are returned to their original colors.
\end{lemma}
\begin{proof}
Only colors of $S$ currently present on outgoing edges need to be removed.
Suppose these are $b_1,\ldots,b_k$, with $0\le k\le2$, on distinct edges
$vw_1,\ldots,vw_k$. There are at least two colors missing at $v$. Choose
$k$ distinct missing colors $t_1,\ldots,t_k$ that do not belong to $S$;
this is possible because at most $2-k$ elements of $S$ are already
missing. For each $i$, release the single color $t_i$ at $w_i$, keeping
$vw_i=b_i$ fixed. These procedures occur in disjoint descendant subtrees
and succeed simultaneously outside an event of probability at most
$k\beta^{2^{r-2}}$ by Lemma~\ref{lem:release}. Afterwards recolor
$vw_i:b_i\to t_i$ for each $i$. Distinct replacements are missing at $v$,
so no conflict occurs. The descendant paths, followed by these last
changes, can be reversed after all $vw_i$ have regained their original
colors. This gives a compatible release of both colors.
\end{proof}

For any $q\ge\Delta+2$, the single-edge release estimate remains valid:
there are at least two spare colors, and the probability that both
occur at the next vertex only decreases as $q$ increases. The
tuple preparation and simultaneous release likewise remain available because $q-3\ge\Delta-1$.
For $q=\Delta+s$, $s\ge2$, the same single-color release recursion has $s$
repair alternatives, giving a failure bound
$\bigl(\binom{\Delta-1}{s}/\binom{\Delta+s-1}{s}\bigr)^{(s^r-1)/(s-1)}$.
We use only the two-spare case in what follows.

\subsection{Reusable buffers}
The next elementary recoloring fact will also fix the order of preparation
moves, which is needed when a lower subtree must be restored.

\begin{lemma}\label{lem:tuple}
Let $u=(u_1,\ldots,u_d)$ and $v=(v_1,\ldots,v_d)$ be injective tuples
of colors from a palette of size at least $d+2$. They can be joined by
single-coordinate recolorings through injective tuples so that each
coordinate changes at most twice.
\end{lemma}
\begin{proof}
Draw a directed arrow from the present color $u_i$ to the requested color
$v_i$ for each coordinate with $u_i\ne v_i$. Since both tuples are
injective, each color has indegree and outdegree at most one. Every
nontrivial component is a directed path or cycle. On a path, a terminal
requested color is vacant in the current tuple. Recolor the corresponding
coordinate first and continue backwards; every coordinate moves once.
For a cycle, choose a vacant color $z$, recolor one coordinate to $z$,
rotate the remaining coordinates around the cycle, and move that first
coordinate from $z$ to its target. It moves twice and every other
coordinate once. There are always at least two vacant colors, so the
operations can be performed successively without interfering.
\end{proof}

Let $e$ enter $w$, with initial color $a$, and let the outgoing edges
$f_i=wz_i$ initially have distinct colors $c_i$. For $q=\Delta+2$,
choose distinct colors $s,t$ missing at $w$. At each $z_i$ we wish to
make every outgoing edge avoid $\{c_i,s,t\}$. This is combinatorially
possible: the remaining palette has $q-3=\Delta-1$ colors and $z_i$
has at most $\Delta-1$ outgoing edges. Choose a target injective tuple
from this palette and a path supplied by Lemma~\ref{lem:tuple}, with all
choices fixed by a deterministic tie-breaking rule.

For each outgoing edge $g=z_i y$, list the colors assigned to $g$ during
its planned path other than its original color. There are at most two.
Apply Lemma~\ref{lem:simultaneous} at $y$ to remove both of them from the
outgoing edges of $y$, while keeping $g$ fixed. Different $g$ have
disjoint descendant subtrees, and these releases can be completed before
any edge incident to $z_i$ changes. Execute the planned injective-tuple
path at $z_i$. Every update is now legal at both endpoints. Throughout
this later path the descendant preparations stay fixed.

At the resulting star $z_i$, all colors $c_i,s,t$ are absent from the
outgoing edges. Thus $f_i$ can take any of these three colors without
changing anything below $z_i$. If the entire procedure ends with
$f_i=c_i$, the original tuple at $z_i$ can be recovered by reversing the
planned path. After every such $g$ returns to its original color, the
releases made below its lower endpoint can also be reversed. Hence the preparation is restored whenever the interface edge returns to its original color.

\begin{lemma}\label{lem:single-buffer}
Suppose all the subtrees below $f_i$ have been prepared as above. Every
legal single-edge recoloring sequence of $e$ at its upper endpoint can be
performed by single-edge updates without changing the prepared lower
stars, except for temporary changes of the $f_i$. Each change of $e$
requires at most four extra moves. Whenever $e$ has its original color
$a$, all $f_i$ have their original colors $c_i$.
\end{lemma}
\begin{proof}
For every current color $x$ of $e$, use the canonical assignment
\begin{equation}\label{eq:canonical-single}
 f_i=\begin{cases}s,&x=c_i,\\c_i,&x\ne c_i.\end{cases}
\end{equation}
The outgoing colors remain distinct and avoid $x$. If neither the old
color $x$ nor the new color $y$ belongs to $\{c_i\}$, change $e$
directly. If $y=c_j$ and $x\notin\{c_i\}$, first recolor $f_j:c_j\to s$
and then $e:x\to c_j$; when $x=s$, use
$f_j:c_j\to t$, $e:s\to c_j$, $f_j:t\to s$ instead. Reverse these paths
when $x=c_i$ and $y\notin\{c_i\}$. Finally, for $x=c_i$, $y=c_j$,
$i\ne j$, use
\[
 f_j:c_j\to t,\quad e:c_i\to c_j,\quad
 f_i:s\to c_i,\quad f_j:t\to s.
\]
In each case the displayed moves remain proper at $w$; at $z_i$ they
are proper because its prepared outgoing star avoids $c_i,s,t$.
The assignment \eqref{eq:canonical-single} returns to the original
$f_i=c_i$ whenever $e=a$, so the same buffer supports an arbitrary
number of interface recolorings.
\end{proof}

Now take $q=\Delta+1$. Again let $e$ enter $w$ in color $a$ and let
$f_i=wz_i$ initially have color $c_i$. Choose a single color $s$ missing
at $w$. The four colors appearing below are specified by the original
star and do not depend on a subsequent interior recoloring path.

\begin{lemma}\label{lem:controller}
For each $i$, one can prepare the subtree below $f_i$, without changing
$f_i$, so that $f_i$ can repeatedly visit $c_i,s,a$, subject to the
constraint at $w$. Every such visit uses at most two neighboring-edge
updates, and returning to $f_i=c_i$ restores the controller to its
prepared state. The preparation requires at most two single-color
releases in distinct descendant subtrees of $z_i$.
\end{lemma}
\begin{proof}
The possibilities depend on whether $a,s$ occur among the outgoing
colors at $z_i$.

If neither occurs, $f_i$ can be recolored directly among $c_i,s,a$.
If $a$ occurs on $g=z_i y$ and $s$ is missing, release $c_i$ at $y$
while holding $g=a$. Then the legal states of $(f_i,g)$ include
\begin{equation}\label{eq:controller-a}
 (c_i,a),\qquad (s,a),\qquad (a,c_i).
\end{equation}
The first two differ by a singleton update of $f_i$; either of them can
be changed into the third by an update of the adjacent pair $(f_i,g)$.
The released color $c_i$ stays absent below $y$ throughout these moves.

If $s$ occurs on $g$ and $a$ is missing, again release $c_i$ below $g$.
The corresponding connected states are
\begin{equation}\label{eq:controller-s}
 (c_i,s),\qquad (s,c_i),\qquad (a,c_i).
\end{equation}
Here the transition from the first state to either of the others changes
both adjacent edges, and the last two differ by a singleton update.

If both $a,s$ occur, write $g_a,g_s$ for their distinct outgoing edges
at $z_i$ and choose a color $t$ missing at $z_i$. First release $t$
below $g_s$, holding $g_s=s$, and then recolor $g_s:s\to t$.
Since $t$ was missing at $z_i$, this is legal. Next release $c_i$
below $g_a$, while keeping $g_a=a$. The color $s$ is now missing at
$z_i$, so \eqref{eq:controller-a} applies with $g=g_a$ and $g_s$
fixed at $t$. The two releases take place in disjoint descendant
subtrees. In all cases the state of every controller edge is determined
by the displayed state and returns to the preparation state when
$f_i=c_i$. After the inner process, the additional change
$g_s:s\to t$ and both release paths are reversed in the opposite order.
\end{proof}

\begin{lemma}\label{lem:pair-buffer}
After preparing all the controllers, any sequence of colors for $e$
that is legal at the upper endpoint can be realized using neighboring-edge
updates. Each prescribed recoloring of $e$ takes a bounded number of
updates. When $e$ returns to $a$, all $f_i$ and their controllers are in
their preparation states.
\end{lemma}
\begin{proof}
When the current color of $e$ is $x$, maintain
\begin{equation}\label{eq:canonical}
 f_i=\begin{cases}s,&x=c_i,\\c_i,&x\ne c_i.\end{cases}
\end{equation}
The controllers are also placed in their unique states corresponding to
these values of $f_i$, as in \eqref{eq:controller-a} or
\eqref{eq:controller-s}. This is a legal configuration at $w$.
If $x,y\notin\{c_i\}$, simply change $e:x\to y$. If
$x\notin\{c_i\}$ and $y=c_j$, move $f_j:c_j\to s$ before recoloring
$e$, unless $x=s$. In that exceptional case, use
\[
 f_j:c_j\to a,\quad e:s\to c_j,\quad f_j:a\to s.
\]
The reversed paths handle $x=c_i$ and $y\notin\{c_j\}$. When
$x=c_i$ and $y=c_j$ for distinct $i,j$, use
\begin{equation}\label{eq:hard-transfer}
 f_i:s\to a,\quad f_j:c_j\to s,\quad
 e:c_i\to c_j,\quad f_i:a\to c_i.
\end{equation}
The intermediate color $a$ is free at $w$ when $e=c_i$, and $s$ becomes
free after the first move. Every change of $f_i$ is realized by its
fixed controller; a neighboring-pair update may change one control edge
simultaneously with $f_i$, but never another edge incident to $w$.
All displayed configurations are proper at both endpoints of every
modified edge. After each prescribed change of $e$, the configuration
is again canonical, so the procedure can be iterated without preparing
anything afresh. At $e=a$, \eqref{eq:canonical} has $f_i=c_i$ for all $i$,
and each controller is back in its preparation state.
\end{proof}

\subsection{Path lifting}
The following connectivity fact concerns only the truncated interior.
It holds with one fixed external parent edge, and applies to either
local chain because neighboring-edge dynamics admits singleton updates.

\begin{lemma}\label{lem:truncated-connect}
For a finite rooted tree with at most one fixed external edge entering
its root, assume that every degree, counting this external edge, is at
most $\Delta$. Its single-edge recoloring graph is connected whenever
$q\ge\Delta+1$.
\end{lemma}
\begin{proof}
Induct on the number of nonfixed edges. Delete a leaf edge $g=uz$,
where $z$ is a leaf, and connect the two restricted configurations by
the induction hypothesis. If a proposed recoloring of an edge at $u$
would use the current color of $g$, move $g$ first to a color different
from the old and new colors of that edge and from all other colors at
$u$. At most $\deg_T(u)\le\Delta$ colors are excluded; the leaf endpoint
places no further restriction. After lifting the path, recolor $g$ to
its target color. This proves connectivity. A fixed parent edge simply counts as an
unmodified additional edge when computing the forbidden colors.
\end{proof}

\begin{proposition}\label{prop:lifting}
Fix all edges at depth $h$ of a rooted tree, and write $M$ for their
original colors and $O$ for all original colors strictly below them.
For either single-edge dynamics with $q\ge\Delta+2$ or neighboring-edge
dynamics with $q\ge\Delta+1$, there is an event $G=G(M,O)$ and constants
$c,C>0$, depending only on $\Delta,q$, such that for every $r\ge r_0$,
\begin{equation}\label{eq:lift-tail}
 \PP(G^c)\le C(\Delta-1)^{h+2}\exp(-c2^{cr}).
\end{equation}
On $G$, every two colorings that both agree with the original coloring
on all edges at depth at least $h$ lie in the same component of the original local chain
restricted to edges through depth $h+r+O(1)$, with all deeper colors
fixed. The component is therefore constant on the entire fiber of upper interior colorings.
\end{proposition}
\begin{proof}
For each vertex $w$ immediately below the cutset, retain its incoming
edge color $a$ and its original outgoing colors $c_i$. For single-edge updates, and also for neighboring-edge updates when
$q\ge\Delta+2$, choose the same $s,t$ at $w$ for the entire construction. For each
$z_i$ plan the fixed injective-tuple path described before
Lemma~\ref{lem:single-buffer}. Every demanded future color for each
outgoing edge $g=z_i y$ is determined by the original pair $(M,O)$.
At $y$, release the set of at most two future colors by
Lemma~\ref{lem:simultaneous}. Only after all these procedures succeed
do we recolor the outgoing tuple at $z_i$. The parent edges $f_i$ remain
unchanged throughout the preparation.

When $q=\Delta+1$ for the neighboring-edge chain, determine from the
original outgoing star at each $z_i$ which of the four cases of Lemma~\ref{lem:controller}
applies. In the nontrivial cases specify the one or two individual
releases below $z_i$ and, when necessary, the change $g_s:s\to t$.
The selection of every controller, its missing colors, and the order of
its operations are fixed from $(M,O)$ before any interior path is
chosen. Different $f_i$ have disjoint descendant subtrees, so these
preparations can be carried out independently.

Let $G$ be the event that all the prescribed descendant releases
succeed within the allocated depth. There are at most
$C_\Delta(\Delta-1)^{h+2}$ prescribed tasks. For the two-spare preparation,
Lemma~\ref{lem:simultaneous} bounds each failure probability by
$C\exp(-c2^{cr})$; for the one-spare preparation use
Lemma~\ref{lem:release} with $k=\lfloor(r-O(1))/2\rfloor$.
The conditional bounds are valid given the original relevant stars;
a union bound gives \eqref{eq:lift-tail}. In particular $G$ depends
only on the original cutset and lower exterior, not on any coloring
strictly above that cutset.

Fix any two interior colorings with those same cutset and lower colors.
By Lemma~\ref{lem:truncated-connect} their restrictions through depth
$h$ can be joined by a finite sequence of single-edge recolorings,
allowing temporary changes of the cutset edges. Starting from the
prepared exterior, follow this path. A move strictly above the cutset
is unchanged. For a move on a cutset edge, apply
Lemma~\ref{lem:single-buffer} or, when $q=\Delta+1$,
Lemma~\ref{lem:pair-buffer} at its
lower endpoint. Each buffer transition starts and ends in its
canonical state. Thus subsequent uses of that cutset edge require no
new descendant preparation, and distinct cutset edges never share a
controller. This lifts \emph{every} finite path selected after $G$
was determined, irrespective of its length.

At the end, each cutset edge has its original color. The canonical
state of each buffer and controller therefore has its original $f_i$
colors. Reverse the fixed preparations, first restoring the stars at
$z_i$ and then undoing the descendant releases in the two-spare
case, or reversing the controller preparations in the one-spare
case. All reversed moves are legal because their incoming edges and
intermediate colors have been restored in the opposite order. The lower
exterior is now exactly $O$; no edge beyond the allocated depth has
been changed. As the preparation was common to all possible interior
paths, this argument proves connectivity of the whole fiber on the
same event $G$.
\end{proof}

\section{Entropy mixing}\label{sec:entropy}

\subsection{Posterior concentration}
Let $u$ be the uniform law on $[q]$. For a rooted edge, write $M_h$ for
the posterior distribution of its color after observing every edge at
distance $h$ below it, under the free-boundary tree measure. If the $d$
outgoing child edges have posterior messages $m_1,\ldots,m_d$, Bayes' rule
and the injective-tuple law give
\begin{equation}\label{eq:message}
 F(m_1,\ldots,m_d)(a)=
 \frac{\displaystyle\sum_{\substack{c_1,\ldots,c_d\,\mathrm{distinct}\\c_i\ne a}}\prod_i m_i(c_i)}
 {\displaystyle\sum_{b=1}^{q}\sum_{\substack{c_1,\ldots,c_d\,\mathrm{distinct}\\c_i\ne b}}\prod_i m_i(c_i)}.
\end{equation}
For a single nonuniform child message, counting the available colors
for all other children gives the exact identity
\begin{equation}\label{eq:single-msg}
 F(m,u,\ldots,u)(a)=\frac{1-m(a)}{q-1},\qquad
 F(m,u,\ldots,u)-u=-\frac{m-u}{q-1}.
\end{equation}
Since $d\le\Delta-1$ and $(\Delta-1)/(q-1)<1$, smoothness yields
constants $\delta>0$, $\rho<1$, and $K<\infty$, depending only on
$\Delta,q$, such that
\begin{align}
 \|F(m_1,\ldots,m_d)-u\|_\infty
 &\le\rho\max_i\|m_i-u\|_\infty
 &&\text{if all }\|m_i-u\|_\infty\le\delta,
 \label{eq:contract}\\
 \|F(m_1,\ldots,m_d)-u\|_\infty
 &\le\frac{\|m_1-u\|_\infty}{q-1}
   +K\sum_{i=2}^d\|m_i-u\|_\infty
 &&\text{if }\|m_i-u\|_\infty\le\delta\ (i\ge2).
 \label{eq:one-bad}
\end{align}
For the second inequality the first argument may be any probability
vector. The denominator of \eqref{eq:message} stays bounded away from
zero on this compact set by \eqref{eq:single-msg}, after decreasing
$\delta$ if needed.

\begin{lemma}\label{lem:posterior}
There exist $c,C>0$, depending only on $\Delta,q$, such that for every
rooted subtree
\begin{equation}\label{eq:strong-post}
 \PP\bigl(\|M_h-u\|_\infty>Ce^{-ch}\bigr)
 \le C\exp(-e^{ch}).
\end{equation}
The analogous bound holds for the posterior of the full outgoing star,
relative to its uniform conditional law given the incoming edge color.
\end{lemma}
\begin{proof}
First couple boundary observations generated by two different incoming
edge colors $a,b$. Whenever the colors entering a descendant vertex
differ, transpose $a,b$ in its outgoing tuple. At most one child-edge
color differs after this transposition, because a proper tuple contains
at most one occurrence of either color. The discrepancy continues
through that vertex only if one specified color occurs among its at
most $\Delta-1$ outgoing edges, an event of probability at most
$\theta=(\Delta-1)/(q-1)<1$. The coupling gives
\[
 \|\mathcal L(\mathrm{boundary}_h\mid a)
   -\mathcal L(\mathrm{boundary}_h\mid b)\|_{\mathrm{TV}}
 \le\theta^h.
\]
Averaging over the uniform incoming color and applying Bayes' rule
therefore gives $\E\|M_h-u\|_\infty\le C\theta^h$. In particular,
for every fixed $\delta>0$, the worst probability $p_h$ that a
depth-$h$ message lies outside the $\delta$-ball tends to zero
uniformly over rooted subtrees.

Choose $\delta,\rho$ as in \eqref{eq:contract}. There is an integer $L$
with the following deterministic property: if among the messages at the
bottom of an $L$-generation subtree at most one is outside the
$\delta$-ball, then the message at the top belongs to the $\delta$-ball.
Indeed, along the unique possible exceptional branch,
\eqref{eq:one-bad} bounds successive deviations by
\[
 x_{j+1}\le(q-1)^{-1}x_j+K\Delta\delta\rho^j,
 \qquad x_0\le1.
\]
All ordinary sibling branches contract under \eqref{eq:contract}.
Choosing $L$ large makes the displayed upper bound below $\delta$.
A branch that becomes nonexceptional earlier only improves the bound.

Condition on the colors of the finite skeleton above a bottom frontier
of depth $L$. The descendant subtrees attached there are independent.
The probability that any one of their depth-$h$ messages is exceptional
is at most $p_h$, uniformly in its fixed incoming color; color
permutation symmetry ensures that the bound is unaffected by the
particular conditioned colors. If the top message is exceptional, at
least two frontier messages are exceptional. There are at most
$(\Delta-1)^L$ frontier edges, so
\begin{equation}\label{eq:double-rec}
 p_{h+L}\le \binom{(\Delta-1)^L}{2}p_h^2.
\end{equation}
When $\Delta=2$ this says that no exceptional message survives after a
sufficient number of levels. Otherwise, take $h_0$ so large that the
right side of \eqref{eq:double-rec} contracts doubly for each residue
class modulo $L$. Iteration gives
$p_h\le C\exp(-e^{ch})$ for all $h$, after changing constants.

Apply this estimate at an intermediate depth $\lfloor h/2\rfloor$.
Outside an event of the stated doubly exponential probability, every
intermediate message is in the $\delta$-ball. The remaining levels
contract by \eqref{eq:contract}, giving deviation at most
$C\rho^{h/2}$. This proves \eqref{eq:strong-post}. Finally, the
conditional law of the outgoing star is proportional to the product
of the child messages restricted to injective tuples. It is a smooth
function of these messages near $u$, uniformly over all degrees at
most $\Delta$. A union bound gives the same estimate for the star.
\end{proof}

Associate with each vertex $v$ the ordered tuple $X_v$ of colors on its
outgoing edges. Every original edge belongs to exactly one such star.
Given the parent star, the conditional law of $X_v$ is uniform among
injective tuples avoiding the incoming edge color. The law of the
whole collection is therefore a finite-state tree Markov measure. At the
global root there is no incoming edge, and its star is uniform over
injective tuples from $[q]$; the arguments below use the same Markov
separation with this minor change.
The single-edge chain changes one coordinate of a star; a neighboring-edge
update changes two coordinates of one star or one coordinate in each
of two adjacent stars.

Let $B_R(v)$ contain the stars at distances less than $R$ below $v$.
Fix its exterior and connect two of its legal colorings when one of the
original chain's updates changes only edges represented by stars in
$B_R(v)$. Let $\mathcal C_{v,R}$ be the connected-component label of this
restricted graph. Write
\[
 Q_{v,R}f=\E[f\mid X_{V\setminus B_R(v)},\mathcal C_{v,R}].
\]
The associated conditional entropy is
$\Ent_{B_R(v)}^*(f)=\Ent(f\mid X_{V\setminus B_R(v)},\mathcal C_{v,R})$,
with the outer expectation understood whenever it appears in a formula.
Since all updates affect at most two adjacent original edges, their
legality depends only on the updated edges and immediately neighboring
edge colors. Consequently, the component partition is determined by
the block and the first outside layer, even for updates involving two
adjacent star states.

\begin{lemma}\label{lem:comp-post}
Fix the incoming edge color at $v$ and any exterior of the rooted
subtree. Let $W$ consist of all star states outside $B_R(v)$ together
with $\mathcal C_{v,R}$. There exist $c,C>0$, depending only on
$\Delta,q$, such that
\begin{equation}\label{eq:comp-post}
 \PP\left(\max_{a:\PP(X_v=a)>0}
 \left|\frac{\PP(X_v=a\mid W)}{\PP(X_v=a)}-1\right|>Ce^{-cR}\right)
 \le C\exp(-e^{cR}).
\end{equation}
This holds for both chains at their stated palette thresholds.
\end{lemma}
\begin{proof}
Choose a cutset $M$ at depth $h=\lfloor R/2\rfloor$ below $v$, and let
$O$ comprise every edge color strictly below the cutset. The rooted
tree Markov property gives $X_v\perp O\mid M$. The path-lifting
event of Proposition~\ref{prop:lifting}, using the remaining block
depth, is measurable with respect to $(M,O)$, has probability
$1-C\exp(-e^{cR})$, and ensures that the component label has the same
value for every possible coloring above $M$ with $M,O$ fixed.
Furthermore, the part of $W$ outside the rooted subtree has already
been fixed, and all of its remaining coordinates belong to $O$ or are
the component label. Hence on the event $G$,
\begin{equation}\label{eq:fiber}
 \PP(X_v=a\mid M,O,W,G)=\PP(X_v=a\mid M).
\end{equation}
Here the component label is constant on the entire interior fiber.

Write $p_a=\PP(X_v=a)$ and $b_a(M)=\PP(X_v=a\mid M)$. Since $G$
and the cutset-good event in Lemma~\ref{lem:posterior} are functions of
$(M,O)$, conditioning and \eqref{eq:fiber} give for every $t>0$
\[
 |\PP(X_v=a\mid W)-p_a|
 \le t+\PP\bigl(G^c\cup\{|b_a(M)-p_a|>t\}\mid W\bigr).
\]
Markov's inequality shows that the probability that the left side
exceeds $2t$ is at most the probability of the union on the right
divided by $t$. Take $t=Ce^{-cR}$, use the star conclusion of
Lemma~\ref{lem:posterior}, and then take a union bound over the
finitely many possible stars. Their nonzero prior probabilities are
uniformly bounded below in terms of $\Delta,q$. Adjusting $c,C$
proves \eqref{eq:comp-post}.
\end{proof}

\subsection{Entropy factorization}
We give the factorization in a form that allows an update to change
two neighboring star states. Its proof uses the Markov structure of
the measure but keeps the \emph{original} component partition.
Let $p>0$ bound below every positive conditional probability of a star
state given its parent. Such a bound exists for fixed $\Delta,q$.

\begin{proposition}\label{prop:factorization}
Suppose for some fixed $R$ and $\varepsilon>0$, uniformly over all
conditional rooted subtrees,
\begin{equation}\label{eq:EM}
 \Ent\bigl(\E[g\mid X_v]\bigr)
 \le\varepsilon\Ent(g)
 \quad\text{for every }g\ge0\text{ with }g=Q_{v,R}g.
\end{equation}
If $4(R+1)\varepsilon<p$, then
\begin{equation}\label{eq:factorization}
 \Ent_\pi(f)\le
 \frac4{p-4(R+1)\varepsilon}
 \sum_v\E_\pi[\Ent_{B_R(v)}^*(f)]
 \qquad(f\ge0).
\end{equation}
\end{proposition}
\begin{proof}
Fix the colors outside one rooted subtree, and normalize $\E f=1$.
Put $A=X_v$, $Q=Q_{v,R}$, and
$H=\Ent(f)-\Ent(Qf)=\E[\Ent(f\mid W)]$, where $W$ records the block
exterior and component. The chain rule for relative entropy also gives
$H=D(f\pi\Vert(Qf)\pi)$. Data processing followed by Pinsker yields
\[
 \|(f\pi)_A-(Qf\pi)_A\|_1\le\sqrt{2H}.
\]
The assumed contraction \eqref{eq:EM} applied to $Qf$ gives
$\Ent(\E[Qf\mid A])\le\varepsilon\Ent(Qf)$, so a second application
of Pinsker shows
\[
 \|(f\pi)_A-\pi_A\|_1
 \le\sqrt{2H}+\sqrt{2\varepsilon\Ent(Qf)}.
\]
Every nonzero atom of $\pi_A$ is at least $p$. Using
$\Ent(\E[f\mid A])\le\chi^2((f\pi)_A\Vert\pi_A)$,
$\chi^2\le p^{-1}\|(f\pi)_A-\pi_A\|_1^2$, and
$\Ent(Qf)\le\Ent(f)$, we obtain
\begin{equation}\label{eq:root-pre}
 \Ent(\E[f\mid A])\le\frac4pH+\frac{4\varepsilon}{p}\Ent(f).
\end{equation}
The entropy chain rule
$\Ent(f)=\Ent(\E[f\mid A])+\E[\Ent(f\mid A)]$,
inserted into \eqref{eq:root-pre}, yields
\begin{equation}\label{eq:root-step}
 \Ent(\E[f\mid A])
 \le\frac4{p-4\varepsilon}\E[\Ent_{B_R(v)}^*(f)]
   +\frac{4\varepsilon}{p-4\varepsilon}\E[\Ent(f\mid A)].
\end{equation}
The estimate holds uniformly over rooted subtrees.

We spell out the conditional expectations needed to
sum \eqref{eq:root-step}. Write $T_v$ for the descendants of $v$
including $v$, and put
\[
 P_v(f)=\E_\pi\left[
 \Ent_{X_v}\bigl(\E[f\mid X_v,X_{V\setminus T_v}]\bigr)
 \right].
\]
Conditional on the parent star, distinct child subtrees are independent.
The entropy chain rule at the root, followed by tensorization over its
child subtrees and the same argument recursively, gives
\begin{equation}\label{eq:tree-chain}
 \Ent_\pi(f)\le\sum_vP_v(f).
\end{equation}

Fix $v$ and condition on the exterior of $T_v$. Let $U$ be the stars
of depths less than $R$ below $v$, $S$ those at depth exactly $R$,
and $Z$ those at greater depth. Set $g=\E[f\mid U,S]$ in this
conditional probability space. The tower property gives
$P_v(g)=P_v(f)$. The Markov property gives
\begin{equation}\label{eq:separator}
 U\ \perp\ Z\mid S.
\end{equation}
By locality, the component label of the $U$-block is a function of
$(U,S)$: an allowed update entirely in $U$ cannot test any color deeper
than its immediately neighboring outside star $S$. In particular,
this is still true when a neighboring-edge update changes two adjacent
stars inside the block; updates that change a fixed outside edge are
not included in the restricted graph. Conditional on $(S,\mathcal C)$,
\eqref{eq:separator} implies that the reference law of $U$ is
independent of $Z$. Thus, writing $\nu_{S,\mathcal C}$ for this
reference law, conditional Jensen says precisely
\[
 \Ent_{\nu_{S,\mathcal C}}
  \bigl(\E[f(U,S,Z)\mid U,S]\bigr)
 \le\E\bigl[\Ent_{\nu_{S,\mathcal C}}(f(U,S,Z))
   \mid S,\mathcal C\bigr].
\]
Averaging over $S,\mathcal C$ gives
\begin{equation}\label{eq:component-jensen}
 \E_\pi[\Ent_{B_R(v)}^*(g)]
 \le\E_\pi[\Ent_{B_R(v)}^*(f)].
\end{equation}
If a branch ends before depth $R$, it contributes no boundary variable;
the same reasoning applies componentwise.

We also need to compare the $P_w$-terms before and after the averaging.
Fix a descendant $w$ at depth at most $R$ below $v$. Let $Y$ be the
full exterior of $T_w$ and $Y_0$ its retained part, which includes the
parent star of $w$; write $Y_1=Y\setminus Y_0$. Conditional on the
parent star, $X_w$ is independent of $Y_1$. For each possible star
color $a$, put
$\phi_a(Y_0,Y_1)=\E[f\mid X_w=a,Y_0,Y_1]$.
The projection defining $g$ integrates omitted variables both in
$T_w$ and in $Y_1$. By the Markov property, integration over the
omitted part of $T_w$ commutes with conditioning on $Y_1$, so
\[
 \E[g\mid X_w=a,Y_0,Y_1]
 =\E[\phi_a(Y_0,Y_1)\mid Y_0].
\]
The conditional reference law $\mu_w(a\mid Y)$ is
$\mu_w(a\mid Y_0)$, because $Y_0$ contains the parent star. Thus
convexity of entropy, with that same reference law on both sides,
gives
\[
 \Ent_{\mu_w(\cdot\mid Y_0)}
    \bigl(\E[\phi_{\cdot}(Y_0,Y_1)\mid Y_0]\bigr)
 \le\E\bigl[
   \Ent_{\mu_w(\cdot\mid Y_0)}
      (\phi_{\cdot}(Y_0,Y_1))\mid Y_0\bigr].
\]
After integrating over $Y_0$, this is exactly
\begin{equation}\label{eq:Pw-jensen}
 P_w(g)\le P_w(f).
\end{equation}
Likewise, applying the entropy chain rule to the retained Markov tree,
conditional on $X_v$, and tensorizing its child subtrees yields
\begin{equation}\label{eq:descendant-chain}
 \E_\pi[\Ent(g\mid X_v)]
 \le\sum_{\substack{w\in T_v\\1\le d(v,w)\le R}}P_w(g).
\end{equation}

Apply \eqref{eq:root-step} to $g$, average over the exterior of $T_v$,
and combine these inequalities. We obtain
\begin{equation}\label{eq:peeling}
 P_v(f)\le\frac4{p-4\varepsilon}
 \E_\pi[\Ent_{B_R(v)}^*(f)]
 +\frac{4\varepsilon}{p-4\varepsilon}
 \sum_{\substack{w\in T_v\\1\le d(v,w)\le R}}P_w(f).
\end{equation}
Each vertex $w$ has at most $R$ ancestors within distance $R$.
Summing \eqref{eq:peeling} over $v$ and using
\eqref{eq:tree-chain}, the last term can be absorbed when
$4(R+1)\varepsilon<p$. This gives \eqref{eq:factorization}.
\end{proof}

To verify \eqref{eq:EM}, we use a finite-output entropy contraction
bound: if a finite-valued $A$ has
minimum positive atom at least $p$ and
\[
 \PP\left(\max_a\left|
 \frac{\PP(A=a\mid W)}{\PP(A=a)}-1\right|>\eta\right)
 \le\delta,
\]
then, for every nonnegative $g=g(W)$,
\begin{equation}\label{eq:KL}
 \Ent(\E[g\mid A])\le
 C_p\left(\eta^2+\frac1{\sqrt{\log(1/\delta)}}\right)\Ent(g).
\end{equation}
\begin{proof}[Proof of the finite-output entropy bound]
We give a proof of \eqref{eq:KL}. Normalize $\E g=1$, write $H=\Ent(g)$, $p_a=\PP(A=a)$, and $Z_a(W)=\PP(A=a\mid W)/p_a$. Each $Z_a$ has mean one and $|Z_a-1|\le p^{-1}$. Let $B$ be the exceptional event, of probability at most $\delta$, and put $s=\E[g\mathbf1_B]$. Since
\[
 \E[g\mid A=a]-1=\E[(g-1)(Z_a-1)],
\]
Pinsker's inequality gives an upper bound $\eta\sqrt{2H}$ for the contribution from $B^c$. On $B$, Cauchy--Schwarz, $|g-1|=|\sqrt g-1|(\sqrt g+1)$, and the Hellinger bound yield
\[
 \E[|g-1|\mathbf1_B]\le C\sqrt{H(s+\delta)}.
\]
Consequently
\[
 |\E[g\mid A=a]-1|^2\le C_pH(\eta^2+s+\delta),
\quad
 \Ent(\E[g\mid A])\le C_pH(\eta^2+s+\delta).
\]
By the variational formula for relative entropy, for every $\lambda>0$,
\[
 s\le\frac{H+\log(1+\delta(e^\lambda-1))}{\lambda}.
\]
Taking $\lambda=\tfrac12\log(1/\delta)$ gives $s\le2(H+\sqrt\delta)/\log(1/\delta)$. When $H\le\sqrt{\log(1/\delta)}$, substitution gives \eqref{eq:KL}. For larger $H$, the bound $\Ent(\E[g\mid A])\le\log(1/p)$ gives the same conclusion. The finitely many constants depend only on the output alphabet's minimum positive atom $p$.
\end{proof}

Lemma~\ref{lem:comp-post} has $\eta=Ce^{-cR}$ and
$\delta=C\exp(-e^{cR})$, and therefore supplies
$\varepsilon\le Ce^{-cR}$. A fixed $R=R(\Delta,q)$ sufficiently large
satisfies $4(R+1)\varepsilon<p$. Proposition~\ref{prop:factorization}
then proves
\begin{equation}\label{eq:block-end}
 \Ent_\pi(f)\le C_{\Delta,q}
 \sum_v\E_\pi[\Ent_{B_R(v)}^*(f)].
\end{equation}

\section{Mixing bounds}
For fixed $\Delta,q,R$, each rooted star block contains a bounded number of edges. Write $E(B_R(v))$ for the original edges represented by stars in $B_R(v)$. Under any fixed exterior coloring, every nonempty component of its local recoloring graph carries an irreducible reversible heat-bath chain, with positive transition probabilities bounded below in terms of $\Delta,q,R$. Hence a finite-state log-Sobolev comparison bounds its conditional entropy by a constant times the sum of pair conditional variances for pairs wholly contained in the block, including singleton pairs. Consequently,
\[
 \E_\pi[\Ent_{B_R(v)}^*(f)]
 \le K_{\Delta,q,R}\!
 \sum_{\substack{\{e,g\}\subseteq E(B_R(v))\\e=g\text{ or }e\cap g\ne\varnothing}}
 \pi[\operatorname{Var}_{\{e,g\}}(\sqrt f)].
\]
The constant is uniform over the exterior and the component. Each edge pair belongs to at most $O(R)$ rooted star blocks, because its endpoints span at most two adjacent star states. Inserting this estimate into \eqref{eq:block-end} proves
\[
 \Ent_\pi(f)\le C_{\Delta,q}\mathcal D_2(\sqrt f)
 \qquad(q\ge\Delta+1).
\]
For single-edge dynamics, the same argument with the two-spare buffer and singleton moves gives $\Ent_\pi(f)\le C_{\Delta,q}\mathcal D_1(\sqrt f)$ when $q\ge\Delta+2$. When $q\ge\Delta+2$, neighboring-edge heat-bath dynamics also dominates the singleton energy since its update set includes $\{e,e\}$. Thus both claims hold on arbitrary bounded-degree trees.

Both chains choose from $\Theta_\Delta(n)$ local update sets. Their normalized log-Sobolev constants are therefore $\Omega_{\Delta,q}(1/n)$; the proposed-color version of single-edge Glauber differs from single-edge heat bath by at most a factor $q$. Finally, $\pi_{\min}^{-1}=|\Omega_q(T)|\le q^{n-1}$. Entropy decay and Pinsker's inequality give, for every fixed $\Delta,q$,
\[
 t_{\mathrm{mix}}(\varepsilon)
 \le K_{\Delta,q}n\left(\log n+\log\frac1\varepsilon\right).
\]

\paragraph{Uniformity in the number of colors.}
The constant in the mixing-time bound can be made independent of $q$ by a direct coupling for large palettes. Regard edges of $T$ as sites of the line graph $G$, and put $m=|E(T)|$ and $D=2\Delta-2$, so that $G$ has maximum degree at most $D$. Couple two arbitrary colorings with disagreement set $S$ by using the same selected edge and proposed color in the single-edge chain. If the selected edge belongs to $S$, the two copies agree after the update whenever the proposal avoids the forbidden colors in both copies. At most $2D$ colors are forbidden. If the selected edge agrees in both copies, a new disagreement occurs only when the proposal is a color used by a disagreeing neighbor in exactly one copy. Summing these events over sites gives
\begin{equation}\label{eq:high-q-single}
 \E\bigl[|S_{t+1}|\mid S_t\bigr]
 \le \left(1-\frac{1-4D/q}{m}\right)|S_t|.
\end{equation}
In particular, if $q\ge8D$, the contraction factor is at most $1-1/(2m)$.

For the neighboring-edge chain, let $\mathcal B$ consist of the singleton sites and adjacent pairs of $G$, and put $N=|\mathcal B|\le (1+D/2)m$. Couple the conditional heat-bath measures on the selected block $B\in\mathcal B$ maximally. If the two boundary conditions differ at $k$ neighbors of $B$, the conditional measures have total variation distance at most $16k/q$, provided $q\ge2D+2$. For a two-site block, each changed boundary color changes at most four membership constraints on the two allowable-color sets. Their sets of feasible ordered color pairs therefore have symmetric difference at most $4kq$, while each set contains at least $(q-D)(q-D-1)\ge q^2/4$ pairs. The singleton case has a smaller bound. Hence a maximal coupling leaves at most $32k/q$ expected disagreements inside $B$. Every site belongs to a singleton block, and a fixed disagreement is adjacent to at most $D(D+1)$ blocks. Averaging over $\mathcal B$ yields
\begin{equation}\label{eq:high-q-pair}
 \E\bigl[|S_{t+1}|\mid S_t\bigr]
 \le \left(1-\frac{1-32D(D+1)/q}{N}\right)|S_t|.
\end{equation}
When $q\ge64(D+1)^2$, this is at most $(1-1/(2N))|S_t|$. Coupling either chain to a stationary copy gives $t_{\mathrm{mix}}(\varepsilon)\le 2N(\log m+\log(1/\varepsilon))$ in this range, with $N=m$ for the single-edge chain. For fixed $\Delta$, the smaller admissible values of $q$ form a finite set, to which the log-Sobolev bounds above apply. Taking the maximum of their constants proves~\eqref{eq:uniform-mix}.

The two thresholds reflect different ways to store a displaced color. Single-edge updates use two spare colors; neighboring-edge updates replace one of these colors by a short exchange through a second layer. In either case, a fixed exterior preparation supports arbitrarily many interior moves.

The constants are not optimized in $\Delta$. Their present dependence comes from a fixed-radius block and a finite-state log-Sobolev comparison. It would be interesting to obtain polynomial dependence on $\Delta$ by estimating the repair paths directly.

\paragraph{Acknowledgment.} The author chose the entropy-based approach and directed the investigation of component entropy methods, identifying relevant ideas and references, including the work of Sly and Zhang. OpenAI's GPT-6 assisted with exploratory derivations, proof checking, literature review, and manuscript editing. The author is responsible for the mathematical arguments and conclusions.

\end{document}